\documentclass[11pt]{article}

\usepackage{amsmath,amsfonts,amssymb,amsthm,braket}
\usepackage[margin=0.63in]{geometry}
\usepackage{url}
\usepackage{breakurl}
\usepackage[colorlinks,
            linkcolor = blue,
            urlcolor  = blue,
            citecolor = blue,
            anchorcolor = blue]{hyperref}
\usepackage{mathrsfs}
\usepackage{mathtools}
\usepackage{thmtools}
\usepackage{thm-restate}
\mathtoolsset{centercolon}
\usepackage{bm}
\usepackage{bookmark}
\usepackage{svg}
\usepackage[font=small]{caption}
\usepackage{titlesec}
\usepackage{enumitem}
\usepackage[backend=biber,style=ieee, maxnames=4]{biblatex}
\usepackage[utf8]{inputenc}
\usepackage{color,soul}
\usepackage{graphicx, fullpage, dsfont}
\usepackage{authblk}

\AtEveryBibitem{%
  \iffieldundef{doi}{}{\clearfield{url}}%
}

\usepackage{xcolor}
\usepackage{tikz}
\usetikzlibrary{patterns,decorations}
\usetikzlibrary{quantikz}
\usepackage{fullpage}
\usepackage{doi}
\usepackage{pgfplots}
\pgfplotsset{compat=newest}
\pgfplotsset{plot coordinates/math parser=false}
\newlength\figureheight
\newlength\figurewidth

\let\originalleft\left
\let\originalright\right
\renewcommand{\left}{\mathopen{}\mathclose\bgroup\originalleft}
\renewcommand{\right}{\aftergroup\egroup\originalright}

\numberwithin{equation}{section}

\newcommand{\BPP}{\mathsf{BPP }}
\newcommand{\sP}{\mathsf{\#P }}

\newcommand{\PostBQP}{\mathsf{PostBQP }}

\newcommand{\cCP}{\mathsf{coC_=P }}

\newcommand{\NP}{\mathsf{NP }}
\newcommand{\PH}{\mathsf{PH }}

\DeclareMathOperator{\pr}{Pr}

\newcommand{\poly}{\operatorname{poly}}

\newcommand{\remove}[1]{}

\newtheorem{theorem}{Theorem}
\newtheorem{lemma}[theorem]{Lemma}

\newtheorem{corollary}[theorem]{Corollary}
\newtheorem{definition}{Definition}
\newtheorem{remark}{Remark}

\newcommand{\eq}[1]{(\ref{eq:#1})}
\newcommand{\alg}[1]{\hyperref[alg:#1]{Algorithm~\ref*{alg:#1}}}
\newcommand{\defn}[1]{\hyperref[defn:#1]{Definition~\ref*{defn:#1}}}
\renewcommand{\sec}[1]{\hyperref[sec:#1]{Section~\ref*{sec:#1}}}
\newcommand{\thm}[1]{\hyperref[thm:#1]{Theorem~\ref*{thm:#1}}}
\newcommand{\lem}[1]{\hyperref[lem:#1]{Lemma~\ref*{lem:#1}}}
\newcommand{\cor}[1]{\hyperref[cor:#1]{Corollary~\ref*{cor:#1}}}
\newcommand{\prb}[1]{\hyperref[prb:#1]{Problem~\ref*{prb:#1}}}
\newcommand{\fig}[1]{\hyperref[fig:#1]{Figure~\ref*{fig:#1}}}

\title{Average-case hardness of estimating probabilities of random quantum circuits with a linear scaling in the error exponent}
\author{\textsf{Hari Krovi} }
\affil[]{Quantum Engineering and Computing\\Raytheon BBN Technologies, Cambridge, MA}
\date{\today}

\begin{document}

\maketitle
\begin{abstract}
We consider the hardness of computing additive approximations to output probabilities of random quantum circuits. We consider three random circuit families, namely, Haar random, $p=1$ QAOA, and random IQP circuits. Our results are as follows.

For Haar random circuits with $m$ gates, we improve on prior results by showing $\cCP$ hardness of average-case additive approximations to an imprecision of $2^{-O(m)}$. Efficient classical simulation of such problems would imply the collapse of the polynomial hierarchy. For constant depth circuits i.e., when $m=O(n)$, this linear scaling in the exponent is to within a constant of the scaling required to show hardness of sampling. Prior to our work, such a result was shown only for Boson Sampling in \cite{bouland2021noise}. We also use recent results in polynomial interpolation to show $\cCP$ hardness under $\BPP$ reductions rather than $\BPP^\NP$ reductions. This improves the results of prior work for Haar random circuits both in terms of the error scaling and the power of reductions.

Next, we consider random $p=1$ QAOA and IQP circuits and show that in the average-case, it is $\cCP$ hard to approximate the output probability to within an additive error of $2^{-O(n)}$. For $p=1$ QAOA circuits, this work constitutes the first average-case hardness result for the problem of approximating output probabilities for random QAOA circuits. The random QAOA circuits we consider include those coming from Sherrington-Kirkpatrick models and Erd\"{o}s-Renyi graphs. For IQP circuits, we prove our results without additional conjectures on the hardness of Ising partition functions. Indeed, a consequence of our results for IQP circuits is that approximating the Ising partition function with imaginary couplings to an additive error of $2^{-O(n)}$ is hard even in the average-case, which extends prior work on worst-case hardness of multiplicative approximation to Ising partition functions.
\end{abstract}

\section{Introduction}\label{sec:introduction}
Quantum computers hold the promise of dramatic speed-ups for certain classes of problems such as factoring \cite{Shor}, Hamiltonian simulation \cite{Low_Chuang} or solving linear systems of equations \cite{HHL} to name a few. However, these gains will likely be realized only with fault-tolerant quantum computers. While these fault-tolerant quantum processors might be years away, the era of small scale devices without fault-tolerance is here and has been termed the NISQ (noisy intermediate scale quantum) era \cite{Preskill2018quantumcomputingin}. NISQ devices are being used to implement heuristic algorithms for problems in optimization, finance, simulations of physical systems and many other areas. Whether there is any speed-up over classical algorithms for these applications still remains to be seen. However, NISQ processors are useful testbeds for certain \emph{quantum} problems such as the recent breakthrough experiment on demonstrating a time crystal \cite{time_crystal}.

Aside from physics simulations, one of the most interesting problems that NISQ devices can excel at, is sampling from output distributions of random quantum circuit families. While this problem may have limited practical use, it is nevertheless a well-defined mathematical problem \cite{aaronson2013computational,Boixo}. Recent ``quantum supremacy" experiments done for different types of random quantum circuits \cite{google_sup,Pan} suggest hardness of sampling from the output distribution by classical algorithms. This experimental evidence has a strong theoretical evidence as well. Starting with the early work of \cite{aaronson2013computational,Bremner2010ClassicalSO}, many techniques have been developed to provide evidence that classical computers will not be able efficiently sample from the output distribution of a random quantum circuit.

Building on that earlier work, \cite{Bouland2018quantum, movassagh2019cayley} showed that a classical algorithm that can exactly compute any output probability of a Haar random quantum circuit with a polynomial number of gates i.e., $m=\poly(n)$ gates would lead to the collapse of the polynomial hierarchy ($\PH$) to a finite level. These results were improved to be robust to additive error in the classical algorithm in \cite{bouland2021noise, movassagh2018efficient} using better polynomial interpolation techniques. This has made the additive error robustness of the earlier results $2^{-O(m\log m)}$, where $m$ is the depth of the quantum circuit. However, to go beyond this scaling using polynomial interpolation, both the Cayley path interpolation and truncated Taylor series (the way it was used in \cite{Bouland2018quantum}) seem insufficient. The Cayley path approach gives a rational function and if we construct a polynomial from it as was done in \cite{KMM21,bouland2021noise}, this gives a scaling of $2^{-O(m\log m)}$ due to the size of the denominator. The Taylor series approach taken in \cite{Bouland2018quantum} gives rise to a degree with an unavoidable factor of $m$, which turns out to be too high. This is because the requirement that we remain close to the random distribution even after a perturbation places a $\log m$ factor in the exponent (using the robust Berlekamp-Welch algorithm in \cite{bouland2021noise}) again leading to $2^{-O(m\log m)}$.

To improve on these results and show hardness for additive error of $2^{-O(m)}$ for a circuit with $m$ gates, we use the interpolation scheme of \cite{Bouland2018quantum}, but then use a truncated Taylor series on the probability, rather than at the level of gates, and obtain a polynomial of degree $O(m/\log m)$. This allows us to use the robust Berlekamp-Welch algorithm (which we show can be improved to give an algorithm in $\BPP$) to give a linear scaling in the exponent of the error. Here, to take the Taylor series, we use the probability as a function of interpolation parameter. The truncated Taylor series polynomial can also be shown to be close enough to the probability of the interpolated circuit. Aside from reducing the degree of the polynomial, this approach also maintains the unitarity of the evolution since each gate after interpolation is still a unitary gate. 

Additionally, as mentioned above, we improve the robust Berlekamp-Welch algorithm developed in \cite{bouland2021noise} to give an algorithm in $\BPP$ rather than $\BPP^\NP$ using recent results in polynomial interpolation \cite{kane2017robust}. In \cite{KMM21}, a $\BPP$ algorithm was given by assuming a stronger condition on the outlier probability. Improving our robustness result to $O(2^{-m})$ would show hardness of sampling for constant depth circuits since $m=O(n)$ for constant depth circuits. However, \cite{Napp_et_al} have shown that one cannot do this by a method agnostic to the circuit depth. In that work, the authors give a classical algorithm that computes output probabilities of Haar random circuits to an additive error of $2^{-n}/2^{n^b}$, where $0<b<1$.

The Quantum Approximate Optimization Algorithm (QAOA) was first introduced in \cite{farhi2014quantum} and has been used as a heuristic algorithm for optimization problems implementable on NISQ devices. Like Boson Sampling and IQP circuits, $p=1$ QAOA also turns out to be universal after post-selection as shown in \cite{farhi2019quantum}. For $p=1$ QAOA, worst-case hardness of multiplicative approximations has also been shown in \cite{farhi2019quantum}. In the context of QAOA, to our knowledge, 
random QAOA circuits have not been studied and this paper is the first to show average-case hardness of approximating output probabilities.

The distribution over QAOA circuits defined here is somewhat related to the distribution of interleaved $X$ and $Z$ diagonal circuits studied in \cite{X_Z_diagonal_t_designs}, where an approximate $t$-design property was shown for the interleaving length large enough. This was shown to be implementable using a gate-set with discrete randomness. The distribution considered here also has randomness only at the level of gates, is assumed to be a continuous distribution and can be more general than a uniform distribution of phases. If the distribution is taken to be uniform, then one can use the results of \cite{Nakata_2014} that uniform distribution of the phases can be implemented using ancillas and a discrete distribution.

The hardness of circuit families such as Instantaneous Polynomial Time (IQP) has also been developed to include robust hardness i.e., hardness of approximating output probabilities or the hardness of sampling. In \cite{Bremner2017achievingquantum}, assuming certain conjectures of polynomials over finite fields, average case hardness of approximating output probabilities has been shown. Random IQP circuits defined in that fashion also need access to three qubit gates such as $CCZ$ gates. However, assuming conjectures about the hardness of Ising partition function estimation, \cite{Bremner2017achievingquantum} shows average-case hardness of approximate sampling using single and two qubit gates. Our results do not assume additional conjectures, but however, need continuous distributions in choosing random gates.

Partition functions are closely related to output amplitudes or probabilities of quantum circuits. For instance, the famous Jones' polynomial can be shown to be the probability amplitude of a quantum circuit \cite{Freedman_et_al, AJL, How_hard}. For IQP circuits, it was shown in \cite{PhysRevLett.117.080501, IQP_Ising} that the output probability amplitudes correspond to Ising partition functions with imaginary coupling constants. Using this connection, in \cite{IQP_Ising}, hardness of worst-case multiplicative approximation of Ising partition function with imaginary coefficients was shown. In the present work, we extend this to average-case additive approximations up to additive error $2^{-O(n)}$.

$\cCP$ hardness of approximating output probabilities is equivalent to the statement that efficient classical algorithm to approximate these probabilities implies the collapse of the $\PH$ (under nondeterministic reductions). In \cite{FGHR}, $\cCP$ hardness of approximating output probabilities of quantum circuits is shown. This result is used in the recent hardness results on Haar random circuit sampling \cite{bouland2021noise, KMM21}. In \cite{fujii_et_al:LIPIcs:2016:6296, Fujii18}, $\cCP$ hardness of the one clean qubit model is studied.

To go beyond asymptotic hardness results, \cite{Dalzell2020howmanyqubitsare,Morimae2020additiveerrorfine} use conjectures in so-called fine-grained complexity to obtain results about the number of qubits needed to show quantum supremacy. Finally, while we do not study it in our paper, an interesting sampling problem called Fourier sampling defined for a particular class of circuits, has been studied in \cite{Fourier_sampling}. The recent survey \cite{rcs_survey} covers these topics including recent experiments and the task of verification of random sampling experiment.

The paper is organized as follows. In \sec{prelims}, some background on random quantum circuits, certain complexity classes and some known results about hardness is provided. In \sec{QAOA}, \sec{haar} and \sec{IQP}, we define the distributions, prove worst-case hardness and bounds on the degrees of polynomials used for interpolation for these random circuit families. In \sec{main}, we prove our main result that additive approximations to a constant fraction of these random circuits is $\cCP$ hard. Finally, in \sec{conclusions}, we present some conclusions and open questions.

\section{Preliminaries}\label{sec:prelims}
In this section, we first describe the structure of the three circuit families we consider. We also discuss some complexity classes that are used later in the paper.

We start with QAOA, focusing on the $p=1$ version. QAOA was first defined in \cite{farhi2014quantum} and was designed for optimization problems. Suppose we have an optimization problem with an objective function $C(z)$, where the goal is to maximize $C(z)$ i.e., find a bit string $z$ for which it is maximized (or approximately maximized). For most interesting problems, the function $C(z)$ is a function of constraints, where each constraint consists of only a few bits of $z$. For this problem, we can construct two unitaries in the $Z$ and $X$ basis labeled $C_Z$ and $C_X$ respectively. They are defined as
\begin{align}
    C_Z&=\sum_zC(z)\ket{z}\bra{z}=\exp(i\sum_{j,k}J_{j,k}\sigma^z_j\sigma^z_k + i\sum_k M_k\sigma^z_k)\,,\\
    C_X&=\exp(i\sum_k\beta_k \sigma^x_k)\,,
\end{align}
where $\sigma^z_k$ and $\sigma^x_k$ are the Pauli $Z$ and $X$ operators acting on the $k^{th}$ qubit.

Using the above two unitaries, QAOA can be defined as a circuit of alternating $Z$ and $X$ diagonal unitaries applied to the equal superposition over all qubits. In the original scheme, the $X$ diagonal term was chosen with a constant $\beta$ (independent of $k$). The phases in these unitaries can be stage dependent. The final state can be written as
\begin{equation}
    \ket{\Bar{\gamma},\Bar{\beta}}=\prod_{i=1}^p \exp(i\beta_iC_X)\exp(i\gamma_iC_Z)\ket{+^n}\,,
\end{equation}
where the state $\ket{+^n}$ is the equal superposition over all qubits. There are many generalizations of QAOA where every aspect from the initial state to the angles and the unitaries may be modified. These generalizations still retain the main aspect of QAOA i.e., alternating non-commuting operations applied to some interesting and easily prepared initial superposition over the qubits. When there is only one stage i.e., $p=1$, the overall circuit is 
\begin{equation}
    U_{QAOA}= C_XC_Z\,.
\end{equation} 
In this paper, we also focus on Haar random circuits and random IQP circuits. For Haar random circuits, there is no special structure except that there are a polynomial number of gates. The circuit is just an arbitrary quantum circuit composed of single and two qubit gates. The distribution of the circuits will be defined later in \sec{haar} and is the same as the one defined in earlier works such as \cite{Bouland2018quantum, movassagh2019cayley}.

The last class of circuits we consider are IQP circuits. This class was first defined in \cite{Bremner2010ClassicalSO} and later developed in \cite{PhysRevLett.117.080501, Bremner2017achievingquantum}. Circuits in this class are somewhat similar to $p=1$ QAOA and can be written in the following form.
\begin{equation}
    U_{IQP}=H^{\otimes n}C_ZH^{\times n}\,,
\end{equation}
where $C_Z$ is an arbitrary $Z$ diagonal unitary with a polynomial number of gates. The post-selected version of this class has been shown to be universal i.e, it is equal to $\PostBQP$.

To define the random circuit version of all these families, we need the definition of an \emph{architecture}. Our definition of an architecture is similar to the one from prior work such as \cite{Napp_et_al}, where a connectivity of qubits and a circuit layout is assumed.
\begin{definition}[Architecture]
Given a connectivity graph with $n$ vertices, we can arrange the qubits on the vertices and think of single qubit gates as being on the vertices and the edges would correspond to two qubit gate locations. A circuit layout is a specification of single and two qubit gate locations i.e., a sequence of vertices and edges (without a specification of the gates themselves). An architecture is defined as the connectivity graph along with a layout of the circuit.
\end{definition}
When we talk about QAOA (respectively Haar random or IQP) architectures later on, we mean architectures that have circuit layouts in the QAOA (respectively Haar and IQP) form.

Next, we briefly discuss some complexity classes that we use in this paper starting with the class $\cCP$. For the definition of other relevant classes such as $\BPP$, $\NP$ and $\sP$, one can refer to standard texts on complexity theory (e.g., \cite{arora2009computational}). An intuitive explanation of the polynomial hierarchy in the context of QAOA is also given in \cite{farhi2014quantum}. The class $\cCP$ consists of problems where one has to decide whether an efficiently computable function $f:\{0,1\}^n\rightarrow \{-1,1\}$ satisfies
\begin{equation}
    \sum_{x\in\{0,1\}^n} f(x)\neq 0\,.
\end{equation}
A reason that the class $\cCP$ is of interest is the result from \cite{T89} that $\mathsf{P}^\sP\subseteq \NP^\cCP$. This implies that efficient classical simulation of a $\cCP$ hard problem collapses the polynomial hierarchy. In \cite{FGHR}, it was shown that deciding whether the output probability of a quantum circuit is non-zero is complete for this class. The circuit construction in that paper is used in recent papers \cite{KMM21, bouland2021noise} on the hardness of Haar random circuits. In this work, we adapt this circuit to QAOA and IQP to show $\cCP$ hardness of additive approximations of output probabilities of $p=1$ QAOA and IQP circuits.

In the next three sections, we introduce the three families of random circuits namely, $p=1$ QAOA, Haar random, and IQP circuits, prove worst-case hardness and certain properties of the interpolated probability.

\section{Random QAOA circuits}\label{sec:QAOA}
The random QAOA circuits we consider are alternating random $Z$ and $X$ basis unitaries on all the qubits. Since we focus on $p=1$, these circuits have a $Z$ basis unitary circuit acting on $\ket{+}^{\otimes n}$ followed by one layer of gates forming an $X$ basis unitary. The measurement is done in the computational basis. To be more precise, we define a random QAOA circuit as follows.
\begin{definition}\label{defn:QAOA_dist}
Given a QAOA architecture, a random QAOA circuit with $p=1$ consists of a random $Z$ basis unitary on $n$ qubits followed by a random $X$ basis unitary. The random $Z$ basis unitary is composed of a sequence of $m_z$ random single or two qubit gates, where the $j^{th}$ gate is of the form
\begin{equation}
    C^{(j)}_Z=\sum_{k_j}e^{i\phi^z_{k_j}}\ket{k_j}\bra{k_j}\,,
\end{equation}
where $k_j$ runs over a basis vectors of $N$ qubits in the $Z$ basis (the index $j$ is fixed by the gate under consideration). Here $N\in\{1,2\}$ depending on whether the $j^{th}$ gate is a single or two qubit gate. The phase $\phi^z_{k_j}$ can be from any continuous distribution.

The $X$ basis unitary is a single layer of $m_x=n$ single qubit unitary gates (where $n$ is the number of qubits) of the form
\begin{equation}
    C^{(j)}_X=\sum_{k_j}e^{i\phi^x_{k_j}}\ket{k_j}\bra{k_j}\,,
\end{equation}
where $k_j$ runs over $\{0,1\}$ and $\phi^x_{k_j}$ are independent and uniformly distributed in $[0,2\pi)$.
\end{definition}
\begin{remark}
We give some examples of distributions of the QAOA phases for the $Z$ gates.
\begin{enumerate}
    \item \emph{Uniform distribution of phases:} The phases $\phi_{k_j}^z$ of the $Z$ basis unitary can be independent and uniformly distributed in the interval $[0,2\pi)$. 
    \item \emph{Sherrington-Kirkpatrick model \cite{SK_QAOA,Montanari}:} $\phi_{k_j}^z$ come from the model where the $Z$ basis unitary is of the form
    \begin{equation}
        C_Z=\exp(i\frac{1}{\sqrt{n}}\sum_{k<\ell} J_{k,\ell}\sigma^z_k\sigma^z_\ell)\,,
    \end{equation}
    where the sum runs over the set of connected qubits $k$ and $\ell$. Here $J_{k,\ell}$ are standard normal distributed $\mathcal{N}(0,1)$ or another continuous distribution. 
    \item \emph{Weighted Max-Cut on Erd\"{o}s-Renyi graphs:} This is a weighted version of the example considered in \cite{SK_QAOA, QAOA_fixed_parameters}. Here $\phi_{k_j}^z$ come from a Max-Cut cost function on a random Erd\"{o}s-Renyi graph $G$. We can center the distribution by defining
    \begin{equation}
        J_{k,\ell}=\frac{\text{Adj}_{k,\ell}-\mathbb{E}(\text{Adj}_{k,\ell})}{\sqrt{\text{Var}(\text{Adj}_{k,\ell})}}\,,
    \end{equation}
    where Adj is the weighted adjacency matrix of the graph $G$ and Var is the variance of a given entry in Adj. To make the overall distribution continuous, the weights must be chosen from a continuous distribution.
\end{enumerate}
\end{remark}

Next, we define a parameterized circuit, (which will be used to interpolate between the worst case QAOA circuit and a random one).
\begin{definition}\label{defn:QAOA_interp}
Suppose we have a QAOA circuit $C=C_XC_Z$ (the worst case circuit), where $C_X$ and $C_Z$ are the $X$ and $Z$ basis unitaries that are composed of a total of $m$ gates defined as follows.
\begin{align}
    C^{(j)}_X&=\sum_{k_j}\exp(i h^x_{k_j})\ket{x_{k_j}}\bra{x_{k_j}}\,,\\
    C^{(j)}_Z&=\sum_{k_j}\exp(i h^z_{k_j})\ket{z_{k_j}}\bra{z_{k_j}}\,,
\end{align}
where $h^x_{k_j}$ and $h^z_{k_j}$ are arbitrary real numbers in $[0,2\pi)$ (e.g., coming from the worst-case instance). Here $x_{k_j}$ and $z_{k_j}$ are basis vectors in the $X$ and $Z$ bases. We consider the following interpolation (with the parameter $\theta$) between each gate in this circuit and a random gate.
\begin{align}
    C^{(j)}_X(\theta)&=\sum_{k_j}\exp( i( h^x_{k_j} + (1-\theta/m)\phi^x_{k_j}))\ket{x_{k_j}}\bra{x_{k_j}}\,,\\
    C^{(j)}_Z(\theta)&=\sum_{k_j}\exp(i(h^z_{k_j}+(1-\theta/m)\phi^z_{k_j}))\ket{z_{k_j}}\bra{z_{k_j}}\,,
\end{align}
where $\phi^x_{k_j}$ and $\phi^z_{k_j}$ are the random phases in $[0,2\pi)$ as in \defn{QAOA_dist}. The $X$ basis unitary is the product of the $C^{(j)}_X(\theta)$ and the $Z$ basis unitary is the product of the $C^{(j)}_Z(\theta)$, where each gate unitary is assumed to act as the identity outside its support.

The full interpolated circuit is $C(\theta)=C_X(\theta)C_Z(\theta)$, where we can write
\begin{equation}\label{eq:C(theta)}
    C(\theta)=\sum_{k,k'}\exp(ih_{k,k'})\exp(i(1-\theta/m)\phi_{k,k'})X_{k}Z_{k'}\,,
\end{equation}
where
\begin{equation}\label{eq:h_phi}
    \phi_{k,k'}=\sum_{i=1}^{m_x}\phi_{k_i}^x+\sum_{i=1}^{m_z}\phi_{k'_i}^z\,\text{ and }\,h_{k,k'} = \sum_{i=1}^{m_x} h^x_{k_i}+\sum_{i=1}^{m_z}h^z_{k'_i}\,,
\end{equation}
and 
\begin{equation}
    X_k=\ket{x_k}\bra{x_k}\,\text{ and }\,Z_k'=\ket{z_k'}\bra{z_k'}\,.
\end{equation}
\end{definition}
\begin{remark}
Note that with the above interpolation, we have a unitary evolution and that $C(m)=C$ i.e., at $\theta=m$, we get the worst-case circuit and when $\theta=0$, we have a random circuit. However, the expression for the probability of outcome $0^n$ is not a polynomial in $\theta$. To get a polynomial in $\theta$, we use standard results in approximation theory later in \thm{polynomial}.
\end{remark}

We now define $p(\theta)$ as the probability of getting the outcome $0^n$ i.e.,
\begin{equation}
    p(\theta)=|\bra{0^n}C(\theta)\ket{+^n}|^2\,.
\end{equation}

The result below puts the above expression into a specific form that will be useful later in \thm{polynomial}.
\begin{lemma}\label{lem:p_theta_QAOA}
For the interpolation defined in \defn{QAOA_interp}, the probability $p(\theta)$ can be written as follows.
\begin{equation}
    p(\theta)=\sum_{r}e^{-i(\theta/m)\Delta \phi_r)}A_r\,,
\end{equation}
where $|A_r|= 1/2^{3n}$ and $|\Delta\phi_r|/m=O(1)$ and $r$ runs over a set of size $2^{4n}$.
\end{lemma}
\begin{proof}
Using the expression for the interpolated circuit in \eq{C(theta)}, this can be written as 
\begin{equation}
    p(\theta)=\sum_{r=(k,k',\ell,\ell')}\exp(i(\Delta h_r+(1-\theta/m)\Delta \phi_r))\bra{0^n}X_{k}Z_{k'}\ket{+^n}\bra{+^n}Z_{\ell'}X_{\ell}\ket{0^n}\,,
\end{equation}
where each element of the tuple $r=(k,k',\ell,\ell')$ runs from $0$ to $2^n$ and
\begin{equation}
    \Delta h_r=(h_{k,k'}-h_{\ell,\ell'})\,,\hspace{0.3in}\Delta\phi_r=(\phi_{k,k'}-\phi_{\ell,\ell'})\,.
\end{equation}
From the definition of $\Delta\phi_r$ above and \eq{h_phi}, we have $|\Delta\phi_r|/m\leq 2\pi$ since we have $\phi_{k_i}^{x}\leq 2\pi$ and $\phi_{k_i}^{z}\leq 2\pi$ for each $i$. Now defining
\begin{equation}
    A_r=e^{i(\Delta h_r+\Delta \phi_r)}\bra{0^n}X_{k}Z_{k'}\ket{+^n}\bra{+^n}Z_{\ell'}X_{\ell}\ket{0^n}\,,
\end{equation}
we also see that $|A_r|=1/2^{3n}$ since
\begin{equation}
    \bra{0^n}X_{k}Z_{k'}\ket{+^n}=\frac{1}{2^n\sqrt{2^n}}\,.
\end{equation}
\end{proof}

Next we prove a hiding lemma for the distribution considered here. It shows that one can focus on the probability of getting $0^n$ outcome in proving hardness, similar to the ones proved in the case of random circuit sampling \cite{Bouland2018quantum}, \cite{movassagh2019cayley}, IQP circuits \cite{Montanaro_2017} and boson sampling \cite{aaronson2013computational}.
\begin{lemma}\label{lem:hiding_QAOA}
Define $p_z$ to be the probability of obtaining some $Z$ basis outcome string $z$ on $n$ qubits.
\begin{equation}
    p_z = |\bra{z}C\ket{+^n}|^2\,,
\end{equation}
where $C$ is a random QAOA circuit. Then, we have that for a given $z$,
\begin{equation}
    \pr_{C}[p_z] = \pr_C[p_0]\,,
\end{equation}
where, $p_0$ is the probability of getting the outcome $0^n$.
\end{lemma}
\begin{proof}
Given a bit string $z$, we can create a new Hamiltonian $H'$ defined as
\begin{equation}
    H'=\sum_k H_k = \sum_{k\in I_z} \frac{1}{4}\sigma^x_k\,,
\end{equation}
where the index set $I_z$ contains bits where the string $z$ has a one. In the above equation, $H_k$ is either $\sigma^x$ or the identity depending on whether $k\in I_z$ or not, respectively. This set of gates is applied along with $C_X$. Applying this does not change the distribution of the random circuit since the phases in the $X$ basis are picked uniformly random. To see this, observe that the $X$ part of the circuit is now
\begin{equation}
   \exp(iH')C_X = \exp(\sum_k 2\pi i (h^x_k+H_k + \phi_k^x) )=\exp( \sum_k 2\pi i (h^x_k + \phi_k^{'x}))\,,
\end{equation}
where $h^x_k$ and $\phi_k^x$ are from \defn{QAOA_interp} and
\begin{equation}
    \phi_k^{'x} = \phi_k^x + H_k\,.
\end{equation}
This shows that $\phi_k^{'x}$ is also uniform. The result of this is a set of $X$ gates (up to phase) on qubits that have a one in the string $z$. There is an overall phase of $i^{|z|}$ in the amplitude of getting the outcome $0^n$, which disappears in the probability. This gives us $p_z(C)=p(C_z)$ some $C_z$, where $C_z$ is distributed according to the same distribution as $C$.
\end{proof}

\subsection{Worst-case hardness}
In this section, we prove hardness of worst case additive approximations. In \cite{farhi2019quantum}, hardness of worst case multiplicative approximation was proved. In \cite{Dalzell2020howmanyqubitsare}, worst-case hardness of additive approximation shown for $p=1$ QAOA circuits with access to CCZ gates via a reduction from hardness of IQP circuits was shown. Here we remove that condition and show hardness with only single and two-qubit gates in \thm{Ising}.

\begin{theorem}\label{thm:ccp_QAOA}
    There exists a family of $p=1$ QAOA circuits $C$ for every $n$, large enough, such that it is $\cCP$ hard to approximate the probability $p(C)$ of getting the outcome $0^n$, to within additive error of $1/2^{2n+1}$. Equivalently, it is $\cCP$ hard to decide if $p(C)=0$ or if $p(C)\geq 1/2^{2n}$.
\end{theorem}
\begin{proof}
Suppose $\tilde{f}$ is an arbitrary $\cCP$ hard function. Define a new function $f:\{0,1\}^n\rightarrow \{-1,1\}$ as follows.
\begin{equation}
    f(x)=(-i)^{|x|}\tilde{f}(x)\,,
\end{equation}
where $|x|$ is the Hamming weight of the bit string $x$ and $i$ is the imaginary unit. The QAOA circuit can now be defined as
\begin{equation}
    U=\exp(i\beta H_X)C_Z\,,
\end{equation}
applied to the initial state $\ket{+^n}$ with $\beta=\pi/4$. Here
\begin{equation}
    H_X=\sum_{j=1}^n \sigma^x_j\,,
\end{equation}
and $C_Z$ is a $Z$ diagonal unitary that applies phases $f(z)$ for each $Z$ basis state $\ket{z}$. In the next result (\thm{Ising}), we show how to implement this as an Ising interaction. Now, applying $U$ to the initial state gives
\begin{equation}
    |\bra{0^n}U\ket{+^n}|^2 = \Bigg|\frac{1}{\sqrt{2^n}}\sum_z \bra{0^n}\exp(i\beta H_X)f(z)\ket{z}\Bigg|^2\,.
\end{equation}
Since $\beta=\pi/4$, we have
\begin{equation}
    \exp(-i\beta \sigma^x)\ket{0}=\cos\beta \ket{0} - i\sin\beta\ket{1}=\frac{1}{\sqrt{2}}(\ket{0}-i\ket{1})\,.
\end{equation}
Using this, we get
\begin{align}
    |\bra{0^n}U\ket{+^n}|^2&=\Bigg|\frac{1}{\sqrt{2^n}}\sum_x (\cos\beta)^{n-|x|}(-i\sin\beta)^{|x|}f(x)\Bigg|^2\\
    &=\Bigg|\frac{1}{2^n}\sum_x (-i)^{|x|}f(x)\Bigg|^2\\
    &=\frac{1}{2^{2n}}\Bigg|\sum_x\Tilde{f}(x)\Bigg|^2\,.
\end{align}
If we have a classical algorithm $\mathcal{O}$ that takes as input the QAOA circuit $U$ and outputs the probability $p(0^n)$ up to an additive error of $1/2^{2n+1}$, then we can decide if $p(0^n)=0$ or if $p(0^n)\geq 1/2^{2n}$. This means that we can decide if $\sum_x \Tilde{f}(x)$ is non-zero, which is a $\cCP$ hard problem.
\end{proof}

\begin{theorem}\label{thm:Ising}
    The operator $C_Z$ in \thm{ccp_QAOA} above can be chosen to be an Ising interaction of the form
    \begin{equation}
        C_Z=\exp(iH_Z)\,,
    \end{equation}
    where
    \begin{equation}
        H_Z=\sum_j b_j\sigma^z_j + \sum_{j,k}c_{j,k}\sigma^z_j\sigma^z_k \, \text{ with } b_j, c_{j,k}\in \mathbb{R}\,.
    \end{equation}
\end{theorem}
\begin{proof}
From the definition of $\cCP$, the function $\Tilde{f}(z):\{0,1\}^n\rightarrow \{-1,1\}$ in \thm{ccp_QAOA} is classically efficiently computable. Using standard techniques in quantum circuit synthesis \cite{nielsen2000quantum}, we can convert this function into a quantum circuit $U_{\tilde{f}}$ that does the following when acting on a basis state.
\begin{equation}
    U_{\tilde{f}}\ket{z}=\tilde{f}(z)\ket{z}\,.
\end{equation}
The circuit implementation of $U_{\tilde{f}}$ can be done using the universal gate set consisting of Hadamard, $CZ$, $S$ and $T$ gates. All these gates except the Hadamards are $Z$ or $ZZ$ rotations and fit the Ising form. To move the Hadamards to ancillas, we use the trick described in \cite{Bremner2010ClassicalSO} for IQP circuits and adapted to QAOA in \cite{farhi2019quantum} and \cite{Dalzell2020howmanyqubitsare}. We briefly describe it here. The gadget from \cite{Bremner2010ClassicalSO} essentially teleports the Hadamard gate from a data qubit onto an ancilla qubit using a $CZ$ gate to entangle them (and the ancilla qubit becomes the new data qubit). The gadget takes an ancilla in the $\ket{+}$ state and applies $CZ$ between them. Then, a Hadamard gate is applied to the data qubit and measured. If we post-select on the outcome $\ket{0}$, then one can see that the state of the ancilla is a Hadamard applied to state of the original data qubit. This is the original scheme for IQP circuits. The modification for QAOA from \cite{farhi2019quantum} and \cite{Dalzell2020howmanyqubitsare} is as follows.

Suppose the state of $n$ qubits is $\ket{\psi}=\ket{0,\psi_0} + \ket{1,\psi_1}$, where the first register is a single qubit and the second register contains the rest of the qubits. Suppose we need to apply a Hadamard gate on the first qubit. Define the single qubit gate $\Tilde{H}$ as follows. 
\begin{equation}
    \Tilde{H}= \exp(-i\frac{\pi}{4}\sigma^x)\,.
\end{equation}
Define the two-qubit operator $Q$ as
\begin{equation}
    Q=\begin{pmatrix}
    1&0&0&0\\
    0&i&0&0\\
    0&0&1&0\\
    0&0&0&-i
    \end{pmatrix}\,.
\end{equation}
Using these, the modified Hadamard gadget is given by the circuit in \fig{hadamard_gadget}. Implementing the circuit gives the state (the gates above the arrows act on the ancilla qubit and the first qubit)
\begin{equation}
    \ket{+,\psi}\substack{Q\\\xrightarrow{\hspace*{0.5cm}}}\, \frac{1}{\sqrt{2}}(\ket{0,0,\psi_0} + \ket{1, 0,\psi_0} + i\ket{0,1,\psi_1} -i\ket{1,1,\psi_1} \substack{\bra{0}(I\otimes\Tilde{H})\\\xrightarrow{\hspace*{1.2cm}}}\,\frac{1}{\sqrt{2}}(\ket{+,\psi_0} + \ket{-,\psi_1})
\end{equation}
The final state can be seen to be $H\ket{\psi}$ where the ancilla qubit becomes the first data qubit (the first data qubit is measured and post-selected on outcome $\ket{0}$).
\begin{figure}[t]
 \centering
     \begin{quantikz}[column sep=1cm,row sep=0.5cm]
     \lstick{$\ket{+}$} & \ctrl{1} & \gate{S} & \qw \rstick{$H\ket{\psi}$}\\
	\lstick{$\ket{\psi}$} & \control{} & \gate{\Tilde{H}} & \qw \rstick{$\bra{0}$}
		\end{quantikz}
    \caption{The modified Hadamard gadget from \cite{farhi2019quantum} acting on the first qubit of $\ket{\psi}$ and the ancilla. Here the operator $Q$, which is a diagonal unitary $diag(1,i,1,-i)$, is split into a $CZ$ and the phase gate.}
    \label{fig:hadamard_gadget}
\end{figure}
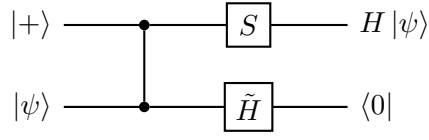
The two qubit gate $Q$ can be implemented as a product of gates $CZ$ and $S=\exp(i\frac{\pi}{4}\sigma^z)$, where $\sigma^z$ acts on the ancilla qubit. To show that $f(z)$ can be expressed as an Ising-like interaction, we need to write $i^{|z|}$ in the same way. This is accomplished by implementing $S^{\otimes n}$ which puts a phase of $i$ if the state is in $\ket{1}$ and $1$ if the state is $\ket{0}$.
\begin{align}
    f(z)=\exp(i\frac{\pi}{4} (\sigma^z_1+\dots +\sigma^z_n))\Tilde{f}(z)\,,
\end{align}
where $\Tilde{f}(z)$ is of the form
\begin{equation}
    \Tilde{f}(z)=\exp(i\sum_{j,k}a_{j,k}\sigma^z_j\sigma^z_k)\,,
\end{equation}
for some coefficients $a_{j,k}$. Putting this together, we see that $f(z)$ can be written as 
\begin{equation}
    f(z)=\exp(i\sum_j b_j\sigma^z_j + i\sum_{j,k}c_{j,k}\sigma^z_j\sigma^z_k)\, \text{ with } b_j, c_{j,k}\in \mathbb{R}\,.
\end{equation}
\end{proof}

\section{Haar random circuits}\label{sec:haar}
In this section, we discuss the hardness of Haar random circuits considered in \cite{Bouland2018quantum, movassagh2018efficient, bouland2021noise, movassagh2019cayley, KMM21}. We recall the definition of Haar random circuit family first.
\begin{definition}\label{defn:distribution_haar}
Given an architecture of $n$ qubits and an arbitrary circuit $C_0$ consisting of single or two qubit gates $G_1,\dots G_m$, a random circuit is defined as replacing each gate $G_i$ by $H_iG_i$, where the gate $H_i$ is drawn from the Haar measure over the unitary group of the same dimension as $G_i$.
\end{definition}
Next we define the interpolation we use for the main result which is similar to the one from \cite{Bouland2018quantum}, but here we use a scale factor of $1/m$ and do not take Taylor series at the level of gates.
\begin{definition}\label{defn:interpolation_haar}
For an architecture, let $G_i$ be some gate sequence. Then for some constant $\theta\in[0,1]$, define the $\theta$ perturbed distribution $\mathcal{H}_{\theta}$ as follows. For each $G_j$, replace it by 
\begin{equation}
    G_j\rightarrow H_j\exp(-\frac{\theta}{m}\log H_j)G_j\,,
\end{equation}
where $H_j$ is a Haar random gate.
\end{definition}
Denote the resulting circuit $C(\theta)$. We show next that the output probability can be written in a specific form as in the case of QAOA.
\begin{lemma}\label{lem:p_theta_haar}
The probability $p(\theta)$ of obtaining the outcome $0^n$ for the interpolated circuit can be written as
\begin{equation}
    p(\theta)=\sum_r e^{-i(\theta/m)\Delta\phi_r}A_r\,,
\end{equation}
where $r$ runs over a set of size $N^{2m}$ and
\begin{equation}
    |A_r|\leq 1 \hspace{0.5in}\text{ and }\hspace{0.5in} \frac{|\Delta\phi_r|}{m}=O(1)\,.
\end{equation}
\end{lemma}
\begin{proof}
Recall that for the circuit $C(\theta)$, each $G_j(\theta)$ is of the form
\begin{equation}
    G_j(\theta)=\exp((1-\theta/m)\log H_j)G_j\,.
\end{equation}
We assume that these single or two-qubit gates are tensored with identity outside their support. For simplicity, we will assume that they are all two qubit gates (possibly acting as the identity on one of the qubits) with dimension denoted $N$. Let $V_j$ be the unitary that diagonalizes $H_j$ and let $\exp(i\phi_{k_j})$ and $\ket{\psi_{k_j}}$ be its eigenvalues and eigenvectors. We can write $G_j(\theta)$ as
\begin{equation}
    G_j(\theta)=V_j^\dag\sum_{k_j=1}^N  e^{i(1-\theta/m)\phi_{k_j}}\ket{\psi_{k_j}}\bra{\psi_{k_j}}V_jG_j=\sum_{k_j}e^{i(1-\theta/m)\phi_{k_j}}\tilde{G}_{k_j}\,,
\end{equation}
where we introduce the notation
\begin{equation}\label{eq:G_tilde}
    \tilde{G}_{k_j} = V_j^\dag \ket{\psi_{k_j}}\bra{\psi_{k_j}}V_jG_j\,.
\end{equation}
We now consider the amplitude $\bra{0}C(\theta)\ket{0}$.
\begin{equation}
    \bra{0}\prod_{j=1}^m G_j(\theta)\ket{0} = \sum_{k_1,\dots ,k_m}e^{i(1-\theta/m)\sum_j\phi_{k_j}} \bra{0}\tilde{G}_{k_1}\tilde{G}_{k_2}\dots \tilde{G}_{k_m}\ket{0}\,.
\end{equation}
The probability can be written as
\begin{equation}
    p(\theta) = \sum_{k,k'}e^{i(1-\theta/m)\Delta\phi_{k,k'}} \bra{0}\tilde{G}_k\ket{0}\bra{0}\tilde{G}_{k'}^\dag\ket{0}\,,
\end{equation}
where $k$ represents the tuple $(k_1,\dots ,k_m)$, similarly $k'$ is $(k'_1,\dots ,k'_m)$, $\tilde{G}_k$ is the product
\begin{equation}\label{eq:G_tilde_k}
    \tilde{G}_k=\tilde{G}_{k_1}\dots \tilde{G}_{k_m}\,,
\end{equation}
and similarly $\tilde{G}_{k'}$ and
\begin{equation}
    \Delta\phi_{k,k'}=(\phi_{k_1}+\dots +\phi_{k_m})-(\phi_{k'_{1}}+\dots+ \phi_{k'_{m}})\,.
\end{equation}
Writing $r=(k,k')$, we see that it runs over a set of size $N^{2m}$. Define
\begin{equation}
    A_r=e^{i\Delta\phi_r}\bra{0}\tilde{G}_k\ket{0}\bra{0}\tilde{G}_{k'}^\dag\ket{0}\,.
\end{equation}
From \eq{G_tilde}, since each $\tilde{G}_k$ is a product of unitaries and projectors, we have
\begin{equation}
    |A_r|\leq 1\,.
\end{equation}
Since each $|\phi_{k_i}|=O(1)$, we have 
\begin{equation}
    \frac{|\Delta\phi_r|}{m}=O(1)\,.
\end{equation}
This proves the lemma.
\end{proof}

The worst case $\cCP$ hardness of additive approximations of arbitrary quantum circuits with a polynomial number of gates is shown in several places \cite{TD04, bouland2021noise, KMM21, FGHR}. We state the result from \cite[Lemma 10]{bouland2021noise}.
\begin{lemma}\label{lem:hardness_haar}
There exist quantum circuits $C$ on $n$ qubits for every $n$ large enough such that it is $\cCP$ hard to decide if $p(C)=0$ or if $p(C)\geq 1/2^{2n}$.
\end{lemma}
In \cite{Bouland2018quantum, movassagh2018efficient}, it was also shown that for Haar random quantum circuits, it is enough to focus on the probability of obtaining $0^n$ where a hiding theorem like the one above for QAOA is shown. We will use these results in \sec{main}. 

\section{Random IQP circuits}\label{sec:IQP}
In this section, we will prove some results for IQP circuits, which will be useful in \sec{main}, where we show our main hardness result for these circuit families. First, we define the random distribution over IQP circuits and the interpolation used here.

\begin{definition}\label{defn:distribution_IQP}
Given an IQP circuit $C$ of the form $H^{\otimes n}C_ZH^{\otimes n}$, where $C_Z$ is a sequence of single or two qubit gates in the $Z$ basis i.e., $C_Z=\prod_i C_i$ we define a random IQP circuit by appending each gate by a random gate of the following form.
\begin{equation}
    \tilde{C}_i=\sum_{z}\exp(i\phi^{(i)}_z)\ket{z}\bra{z}\,,
\end{equation}
where $\ket{z}$ is the $Z$ basis for the gate and $\phi^{(i)}_z$ is a random phase in $[0,2\pi)$. The random IQP circuit can be written as $H^{\otimes n}\prod_i \tilde{C}_iC_iH^{\otimes n}$.
\end{definition}

\begin{definition}
Given an IQP circuit $C$ consisting of $m$ gates, we define the interpolated circuit $C(\theta)$ between $C$ and a random IQP circuit $\tilde{C}$ by creating an interpolation $C_i(\theta)$ for each gate as follows.
\begin{equation}
    C_i(\theta) = \sum_z \exp(ih^{(i)}_z + (1-\theta/m)\phi^{(i)}_z)\ket{z}\bra{z}\,,
\end{equation}
where $h^{(i)}_z$ is the eigenvalue of $C_i$ corresponding to the eigenstate $\ket{z}$ and $\phi^{(i)}_z$ is the eigenvalue of $\tilde{C}_i$.
\end{definition}

We now show the worst-case hardness result.
\begin{theorem}\label{thm:ccp_IQP}
    There exists a family of IQP circuits $C_n$ for each $n$ large enough, such that it is $\cCP$ hard to approximate the probability $|\bra{0}C_n\ket{0}|^2$ to within an additive error of $1/2^{2n+1}$.
\end{theorem}
\begin{proof}
The proof is along the same lines as the one for QAOA (\thm{ccp_QAOA}). We start by considering a $\cCP$ hard function 
\begin{equation}
    f:\{0,1\}^n\rightarrow \{-1,1\}\,.
\end{equation}
Using standard techniques in quantum computing \cite{nielsen2000quantum}, we can then construct a quantum circuit that computes $f(z)$ as a phase i.e.,
\begin{equation}
    U_f=\sum_z f(z)\ket{z}\bra{z}\,.
\end{equation}
Now by a similar calculation as in \thm{ccp_QAOA}
\begin{equation}
    p(C)=|\bra{0^n}H^{\otimes n}U_fH^{\otimes n}\ket{0^n}|^2=\frac{1}{2^{2n}}\Big(\sum_z f(z)\Big)^2\,.
\end{equation}
Since deciding if $\sum_zf(z)$ is positive or zero is $\cCP$ hard, we find that deciding if the probability $p(C)$ is greater than $1/2^{2n}$ or zero is also $\cCP$ hard. The form of the circuit above is not quite in the IQP form since $U_f$ need not consist of only $Z$ diagonal gates. It could contain Hadamard gates as well. We explain below how to make it $Z$ diagonal.

The circuit to implement $U_f$ can be assumed to consist of the standard gate set containing Hadamards, CNOTs, phase and $T$ gates. All gates except the Hadamards are $Z$ diagonal. Now we can use a Hadamard gadget described in \cite{Bremner2010ClassicalSO} to move the Hadamards to ancillas. Recall that the gadget takes an ancilla in the $\ket{+}$ state and applies $CZ$ between them. Then, a Hadamard gate is applied to the data qubit, which is then measured. If we now post-select on the outcome $\ket{0}$, then the state of the ancilla is a Hadamard gate applied to state of the original data qubit. This brings the circuit into the IQP form.
\end{proof}
We need the hiding lemma next.
\begin{lemma}\label{lem:hiding_IQP}
Given a random IQP circuit, define
\begin{equation}
    p_z(C)=|\bra{z}C\ket{0^n}|^2\,,
\end{equation}
where $z$ is an arbitrary bit string in the $Z$ basis. Let $p(C)$ denote the probability of getting the outcome $0^n$. Then we have
\begin{equation}
    \pr_C[p_z(C)]=\pr_C[p(C)]\,,
\end{equation}
where the probability is over the randomness in the circuit distribution.
\end{lemma}
\begin{proof}
Given a bit string $z$ and a random circuit $C$, define a new circuit $C_z$ such that 
\begin{equation}
    C_z=X^{\otimes I}C\,,
\end{equation}
where $X^{\otimes I}$ is a set of $X$ gates on the set $I$ of qubits where $z$ has a one. Since $C$ contains Hadamards on either side, we can commute $X^{\otimes I}$ through the Hadamards to get $Z^{\otimes I}$. This is then absorbed by the random $Z$ basis gates since the distribution is invariant under left multiplication. Therefore, we have
\begin{equation}
    \pr_C[p_z(C)]=\pr_C[p(C_z)] = \pr_C[p(C)]\,.
\end{equation}
\end{proof}

Next we prove that $p(\theta)$ can be written in a specific form that will be useful later.
\begin{lemma}\label{lem:p_theta_IQP}
The probability $p(\theta)$ of an IQP circuit can be written as 
\begin{equation}
    p(\theta)=\sum_r e^{-i(\theta/m)\Delta\phi_r}A_r\,,
\end{equation}
where
\begin{equation}
    |A_r|= 1/2^{2n} \hspace{0.5in} \text{ and }\hspace{0.5in} \frac{|\Delta\phi_r|}{m}=O(1)\,,
\end{equation}
and $r$ runs over a set of size $2^{2n}$.
\end{lemma}
\begin{proof}
Let us write the expression for $p(\theta)$ as follows.
\begin{equation}
    p(\theta)=|\sum_z\bra{0}H^{\otimes n}e^{ih_z+i(1-\theta)\phi_z}\ket{z}\bra{z}H^{\otimes n}\ket{0}|^2\,,
\end{equation}
where $h_z$ and $\phi_z$ are the phases that correspond to the basis vector $\ket{z}$ from all the gates of the worst case circuit and the random circuit respectively. The probability can be re-written as
\begin{equation}
    p(\theta)=\sum_{r=(z,z')}e^{i\Delta h_r+i(1-\theta)\Delta\phi_r}\bra{0}H^{\otimes n}\ket{z}\bra{z}H^{\otimes n}\ket{0}\bra{0}H^{\otimes n}\ket{z'}\bra{z'}H^{\otimes n}\ket{0}\,,
\end{equation}
where $r=(z,z')$ runs over a set of size $2^{2m}$ and
\begin{equation}
    \Delta\phi_{z,z'} = (\phi_z - \phi_{z'})\,\text{ and }\,\Delta h_{z,z'}=(h_z-h_{z'})\,.
\end{equation}
Taking 
\begin{equation}
    A_r=e^{i\Delta h_r+i\Delta\phi_r}\bra{0}H^{\otimes n}\ket{z}\bra{z}H^{\otimes n}\ket{0}\bra{0}H^{\otimes n}\ket{z'}\bra{z'}H^{\otimes n}\ket{0}\,,
\end{equation}
we have $|A_r|= 1/2^{2n}$ and since there are $m$ gates and the contribution to the phase from each gate is $O(1)$, we also have
\begin{equation}
    \frac{|\Delta\phi_r|}{m}=O(1)\,.
\end{equation}
\end{proof}

\section{Main result}\label{sec:main}
In this section, we prove $\cCP$ hardness of the three random quantum circuit families considered above, namely $p=1$ QAOA, Haar random and IQP. Let us denote this set $\mathcal{F}$. First, we show that there exists a low degree polynomial approximating the interpolated probability $p(\theta)$.
\begin{theorem}\label{thm:polynomial}
For each of the families in $\mathcal{F}$ above and for the interpolation defined for each family, suppose that there is an interpolated circuit with $m$ gates with outcome probability $p(\theta)$. Then there exists a polynomial $\tilde{p}(\theta)$ of degree $d$, where
\begin{equation}\label{eq:d_choice}
    d=O(\frac{m}{\log m})\,,
\end{equation}
such that $\tilde{p}(m)=p(m)$ and
\begin{equation}
    |p(\theta)-\tilde{p}(\theta)|\leq \frac{1}{2^{2n+2}}\,,
\end{equation}
for $\theta\leq 1$.
\end{theorem}
\begin{proof}
For some $d$ that will set later, let $d+1$ points $x_0,\dots x_d$ be such that the first $d$ of them are in $[0,1]$ and $x_d=m$. Then using standard approximation theory \cite[Chapter 6]{suli2003introduction}, there exists a degree $d$ polynomial $\tilde{p}(\theta)$ such that
\begin{equation}\label{eq:error}
    |p(\theta)-\tilde{p}(\theta)|\leq \frac{|p^{(d+1)}(\theta_0)|}{(d+1)!}(\theta-x_0)\dots (\theta-x_d)\,,
\end{equation}
where $p^{(d+1)}(\theta_0)$ is the $d+1^{st}$ derivative of $p(\theta)$ evaluated at some $\theta_0\in [0,m]$. We now bound this quantity. 

From \lem{p_theta_QAOA}, \lem{p_theta_haar} and \lem{p_theta_IQP}, we have that the expression for $p(\theta)$ for each of these circuit families is
\begin{equation}
    p(\theta)=\sum_r e^{i(\theta/m)\Delta\phi_r}A_r\,,
\end{equation}
where $r$ runs over a set of size at most $N^{2m}$ for each family (with $N=4$ being the two-qubit dimension). This gives us
\begin{equation}
    p^{(d+1)}(\theta_0)=\sum_{r}(-i\frac{\Delta\phi_r}{m})^{d+1}e^{i(\theta/m)\Delta\phi_r}A_r\,.
\end{equation}
Since $|\Delta\phi_r|/m\leq 2\pi$ and $|A_r|\leq 1$ and $r$ runs over a set of size $N^{2m}$, we have
\begin{equation}
    |p^{(d+1)}(\theta_0)|\leq N^{2m}(2\pi)^{d+1}\,.
\end{equation}
Plugging this back into \eq{error}, we see that in order to make 
\begin{equation}
    |p(\theta)-\tilde{p}(\theta)|\leq \frac{1}{2^{2m+2}}\,,
\end{equation}
for $\theta\leq 1$, we need to choose $d$ so that 
\begin{equation}\label{eq:d_requirement}
    (d+1)!\geq 2^{2m+2}N^{2m}(2\pi)^{d+1}m\,.
\end{equation}
Using the fact that
\begin{equation}
    \ln (d+1)!\geq (d+1)\ln \frac{d+1}{e}\,,
\end{equation}
we have (using the above two equations) that a sufficient condition on $d$ is
\begin{equation}
    (d+1)\ln \frac{d+1}{2\pi e}\geq (6m+2)\ln 2 + \ln m\,.
\end{equation}
Or more simply, another sufficient condition (when $m>4$) is
\begin{equation}
    (2\pi e)d'\ln d'\geq 8m\,,
\end{equation}
where $d'=(d+1)/(2\pi e)$. In other words, $d'\ln d'\geq m$ is sufficient. It is easy to see that when 
\begin{equation}
    d'\geq \frac{2m}{\ln m}\,,
\end{equation}
we have
\begin{equation}
    d'\ln d'\geq m\,.
\end{equation}
To see this, notice that
\begin{equation}
    d'\ln d' \geq \frac{2m}{\ln m}\ln (\frac{2m}{\ln m})=2m - m\frac{\ln(\frac{\ln m}{2})}{\frac{\ln m}{2}}\geq m\,,
\end{equation}
since $(1/2)\ln m\geq \ln ((1/2)\ln m)$. This gives us that when
\begin{equation}\label{eq:d_value}
    d+1=\frac{4\pi e m}{\ln m}\,,
\end{equation}
we have 
\begin{equation}
    (d+1)!\geq 2^{2m+2}N^{2m}(2\pi)^{d+1}m\,.
\end{equation}
Using this and since $x_0,\dots ,x_{d-1}$ are all less than $1$ and $x_d=m$, we have for $\theta\leq 1$ 
\begin{equation}
    |p(\theta)-\tilde{p}(\theta)|\leq \frac{1}{2^{2m+2}}\leq \frac{1}{2^{2n+2}}\,,
\end{equation}
since $m\geq n$. By construction, we also have $\tilde{p}(m)=p(m)$. This proves the theorem.
\end{proof}
For QAOA and IQP circuits, we can improve the dependence on $d$ as given in the following corollary.
\begin{corollary}\label{cor:QAOA_IQP}
For QAOA and IQP circuits, the dependence of $d$ can be improved to
\begin{equation}
    d=O(\frac{n}{\log n})\,.
\end{equation}
\end{corollary}
\begin{proof}
For QAOA and IQP circuits, the requirement on $d$ from \eq{d_requirement} can be modified to
\begin{align}
    &(d+1)!\geq 2^{3n+2}(2\pi)^{d+1}m \text{ for QAOA}\label{eq:conditions_QAOA}\\
    &(d+1)!\geq 2^{2n+2}(2\pi)^{d+1}m \text{ for IQP}\label{eq:conditions_IQP}\,.
\end{align}
Since $m$ is $\poly(n)$, we can take $m\leq O(n^s)$, for some constant $s$. This means that in both cases, we can take 
\begin{equation}
    d=O(\frac{n}{\log n})\,,
\end{equation}
to satisfy the conditions in \eq{conditions_QAOA} and \eq{conditions_IQP}.
\end{proof}

Next, we prove a result on the total variation distance of the perturbed distributions for circuits from $\mathcal{F}$.
\begin{theorem}\label{thm:TVD}
    Given a quantum circuit family from $\mathcal{F}$, let $C$ be a circuit from that family with $m$ gates and let $C(\theta)$ be an interpolated circuit. Let the distribution over all circuits from $\mathcal{F}$ at a value of $\theta$ be $\mathcal{H}(\theta)$. Then
    \begin{enumerate}
        \item $C(m)=C$.
        \item For $\theta\le \Delta$, where $\Delta=O(1/N)$, we have that the total variation distance between the distributions $\mathcal{H}(\theta)$ and $\mathcal{H}(0)$ is a constant.
    \end{enumerate}
\end{theorem}
\begin{proof}
Let the circuit $C(\theta)$ with $m$ gates be 
\begin{equation}
    C(\theta)=G_1(\theta)\dots G_m(\theta)\,,
\end{equation}
where each $G_i(\theta)$ is of dimension $N$. From the definition of interpolation for each family, it is easy to see that $(1)$ in the statement of the theorem. 

For the proof of $(2)$, we follow the proof in \cite{bouland2021noise}. Let $\nu_\theta$ be the probability density function of the distribution of $C_i(\theta)$ and let $\Phi$ denote the set of eigenvalues of $C_i(\theta)$. Since for Haar random unitary gates, the eigenvalues and eigenvectors are independent, and for QAOA and IQP circuits, we only pick eigenvalues at random, let $\nu_\theta$ be an eigenvalue density where 
\begin{equation}
    \nu_\theta(\phi)=\nu(f_\theta^{-1}(\phi))\frac{df_\theta^{-1}(\phi)}{d\phi}\,,
\end{equation}
with $f_\theta(\phi)=(1-\theta/m)\phi$. At $\theta=0$, $\nu_0(\phi)$ comes from the distribution of random gates for each family in $\mathcal{F}$. Now
\begin{equation}
    \frac{df_\theta^{-1}(\phi)}{d\phi}=\frac{1}{1-\theta/m}\,.
\end{equation}
Using this, for $\theta/m\leq 1/N$, we get 
\begin{align}
    |f_\theta^{-1}(\phi)-\phi|\leq O(N^{-1})\,,\hspace{0.3in} \Big|\frac{df_\theta^{-1}(\phi)}{d\phi}-1\Big|\leq O(N^{-1})\,,
\end{align}
since $|\phi|\leq 2\pi$ and $\theta/m\leq 1$.

The total variation distance between the distributions at $\theta=0$ and a nonzero $\theta$ is
\begin{equation}
    D_{TV}(\nu_\theta,\nu_0)=\frac{1}{2}\int |\nu_\theta(\Phi)-\nu_0(\Phi)|\prod_id\phi_i\,.
\end{equation}
For $\theta/m\leq O(N^{-1})$, the integrand can be estimated as follows.
\begin{align}
    &\Big|\nu_\theta(\Phi)-\nu_0(\Phi)\Big|\nonumber\\
    &=\Big|\nu_0(f_\theta^{-1}(\Phi))\prod_i\Big|\frac{f_\theta^{-1}(\phi_i)}{d\phi_i}\Big|-\nu_0(\Phi)\Big|\,,\\
    &\leq \Big|\Big(\nu_0(\Phi)+\sum_i O(N^{-1})\Big)\prod_i(1+O(N^{-1}))-\nu_0(\Phi)\Big|\,,\\
    &\leq O(1)\,.
\end{align}
The integral is also bounded by $O(1)$ since each $|\phi_i|\leq 2\pi$. Since there are $m$ gates, we have that if $\theta/m\leq  O(1/(m N))$ i.e., $\theta\leq O(1/N)$, then $\mathcal{H}(0)$ and $\mathcal{H}(\theta)$ are $O(1)$-close in total variation distance.
\end{proof}

In this section, we prove our main hardness result in \thm{main}. Before we state and prove it, we need the following results. The first one, \thm{RBW_P}, improves on the result proved in \cite{bouland2021noise} stated below.
\begin{theorem}[\cite{bouland2021noise}]\label{thm:RBW}
    For a degree $d$ polynomial $P(x)$, suppose there is a set of points $D=\{(x_i,y_i)\}$ such that $|D|=100d^2$, and $x_i$ are equally spaced in the interval $[0,\Delta]$ ($\Delta<1$). Furthermore, assume that each point satisfies 
    \begin{equation}
        \pr[|y_i-P(x_i)|\geq\delta]\leq\eta\,,
    \end{equation}
    where $\eta<1/4$ is a constant, then there exists a $\mathsf{P}^{\NP}$ algorithm that takes $D$ as input and returns a number $p_1$ such that
    \begin{equation}
        |p_1-P(1)|\leq \delta\exp(d\log\Delta^{-1}+O(d))\,,
    \end{equation}
    with success probability at least $2/3$.
\end{theorem}
We point out below in \thm{RBW_P} that the above result can be improved to give a $\BPP$ algorithm, rather than a $\mathsf{P}^{\NP}$ algorithm. We need the next result proved in \cite{kane2017robust}, which solves the problem of polynomial regression with constant outlier probability with nearly optimal parameters.
\begin{theorem}[\cite{kane2017robust}]\label{thm:robust_poly}
    Suppose $Q(x):[-1,1]\rightarrow \mathbb{R}$ is a degree $d$ polynomial and $D=\{(x_i,y_i)\}$ is a set of points such that $|D|=O(\frac{d}{\epsilon}\log\frac{d}{\eta'\epsilon})$, where $x_i$ are drawn from the Chebyshev distribution. Suppose also that
    \begin{equation}
        \pr[|y_i-Q(x_i)|\geq\delta]\leq\eta\,,
    \end{equation}
    where $\eta<1/2$, then there exists an algorithm that runs in time polynomial in the sample size and finds, with probability at least $1-\eta'$, a degree $d$ polynomial $\hat{Q}$ such that 
    \begin{equation}
        |Q-\hat{Q}|_\infty\leq (2+\epsilon)\delta\,.
    \end{equation}
\end{theorem}
Using this theorem, we can prove the following.
\begin{theorem}\label{thm:RBW_P}
    For a degree $d$ polynomial $P(x)$, suppose there is a set of points $D=\{(x_i,y_i)\}$ such that $|D|=O(d\log d)$ and $x_i\in [0,\Delta]$ are such that $(2x_i/\Delta)-1$ are Chebyshev distributed in the interval $[-1,1]$, where $\Delta<1$. Suppose also that each point satisfies 
    \begin{equation}
        \pr[|y_i-P(x_i)|\geq\delta]\leq\eta\,,
    \end{equation}
    where $\eta<1/4$ is a constant, then there exists a $\BPP$ algorithm that takes $D$ as input and returns a number $p_m$ such that
    \begin{equation}
        |p_m-P(m)|\leq \frac{9\delta}{4}\frac{(8m)^d}{\Delta^d}=\delta'\exp(-d\log \Delta')\,,
    \end{equation}
    with success probability at least $2/3$, where
    \begin{equation}
        \Delta'=\frac{\Delta}{8m}\hspace{0.2in}\text{and}\hspace{0.2in}\delta'=\frac{9\delta}{4}\,.
    \end{equation}
\end{theorem}
\begin{proof}
The proof essentially puts together \thm{robust_poly} above and \cite[Lemma 8]{bouland2021noise}. First pick $Q(x)=P(\frac{x+1}{2}\Delta)$ so that $Q(x):[-1,1]\rightarrow \mathbb{R}$ is a degree $d$ polynomial. Define a new set of points $D'=\{(x'_i,y_i)\}$, where $x'_i=(2x/\Delta)-1$. The points $x'_i$ are Chebyshev distributed $\mathcal{D}\sim\frac{1}{\sqrt{1-x_i{'^2}}}$. Pick $D$ such that $|D|=O(\frac{d}{\epsilon}\log\frac{d}{\eta'\epsilon})$, where $\eta'=1/3$ and $\epsilon=1/4$. With respect to the set of points $D'$, the polynomial $Q(x)$ satisfies $Q(x):[-1,1]\rightarrow \mathbb{R}$ and 
\begin{equation}
    \pr[|y_i-Q(x'_i)|\geq \delta]\leq\eta\,.
\end{equation}
Therefore, using \thm{robust_poly}, we can find a degree $d$ polynomial $\hat{Q}(x)$ with probability $1-\eta$ in polynomial time such that 
\begin{equation}
    |Q-\hat{Q}|_\infty\leq \frac{9\delta}{4}=\delta'\,.
\end{equation}
Now define $p_m=\hat{Q}(\frac{2m}{\Delta}-1)$. Since we have $P(m)=Q(\frac{2m}{\Delta}-1)$, we get
\begin{equation}
    |p_m-P(m)|=|R(\frac{2m}{\Delta}-1)|\,,
\end{equation}
where $R(x)=Q(x)-\hat{Q}(x)$. Since $|R(x)|\leq \delta'$ for $x\in [-1,1]$, using \cite[Lemma 4.1]{Sherstov}, we have 
\begin{equation}
    \sum_{i=0}^d |a_i|\leq 4^d\delta'\,,
\end{equation}
where $a_i$ are the coefficients of $R(x)$. Now we have
\begin{align}
    |p_m-P(m)|&=|R(\frac{2m}{\Delta}-1)|\\
    &\leq \sum_{i=0}^d |a_i|(\frac{2m}{\Delta}-1)^i\\
    &\leq \sum_{i=0}^d |a_i|(\frac{2m}{\Delta})^i\\
    &\leq \frac{\delta'}{\Delta^{'d}}\,,
\end{align}
which proves the theorem.
\end{proof}
The next result proved in \cite{Bouland2018quantum} allows us to use a distribution close to a random one. We include a proof here for completeness.
\begin{lemma}[\cite{Bouland2018quantum}]\label{lem:random_close}
Suppose $\mathcal{D}$ and $\mathcal{D}'$ are two distributions such that the total variation distance between them is upper bounded by a constant $\delta$. Let $p(C)$ be the probability of obtaining outcome $0^n$ when measuring the quantum circuit $C$. If there exists a classical algorithm $\mathcal{O}$ such that 
\begin{equation}
    \pr_{C\sim\mathcal{D}}[|\mathcal{O}(C)-p(C)|\geq\epsilon]\leq \eta\,,
\end{equation}
then we also have
\begin{equation}
    \pr_{C\sim\mathcal{D}'}[|\mathcal{O}(C)-p(C)|\geq\epsilon]\leq \eta+\delta\,.
\end{equation}
\end{lemma}
\begin{proof}
Let a circuit $C$ be drawn from $\mathcal{D}$. Let $E$ be the event that
\begin{equation}
    |\mathcal{O}(C)-p(C)|\leq \epsilon\,.
\end{equation}
Since the total variation distance between $\mathcal{D}$ and $\mathcal{D}'$ is upper bounded by $\delta$, we have
\begin{equation}
    |\pr_{C\sim\mathcal{D}}[E]-\pr_{C\sim\mathcal{D}'}[E]|\leq \delta\,.
\end{equation}
This gives us
\begin{equation}
    \pr_{C\sim\mathcal{D}'}[|\mathcal{O}(C)-p(C)|\geq\epsilon]\leq \eta+\delta\,.
\end{equation}
\end{proof}

We now give our main result.
\begin{theorem}\label{thm:main}
    It is $\cCP$ hard to approximate the output probability of a random circuit with $m$ gates from any of Haar random, $p=1$ QAOA or IQP circuit families to within an additive error of $\epsilon/2^{c m}$ for at least $1-\eta$ fraction of circuits, where $\epsilon$, $c$ and $\eta$ are constants.
\end{theorem}
\begin{proof}
Let $\tilde{C}$ be a random circuit from any the above families. Using the hiding property of the distribution, we can focus on the probability of obtaining the outcome $0^n$ (denoted $p(\tilde{C})$). Let $\mathcal{O}$ be a classical algorithm that can take as input a classical description of a random quantum circuit and outputs $\mathcal{O}(\tilde{C})$ such that
\begin{equation}\label{eq:O_prob}
    \pr[|\mathcal{O}(\tilde{C}) - p(\tilde{C})|\geq\delta]\leq \eta\,,
\end{equation}
where $\eta$ is a constant and $\delta$ such that 
\begin{equation}\label{eq:delta}
    \delta=\frac{4}{9}\frac{\Delta^{'d}}{2^{2n+2}}\,,
\end{equation}
where $d=O(m/\log m)$, $\Delta'=\Delta/8m$, $\Delta=O(1/N)$, $N$ is the qudit dimension. All this makes $\delta\leq\epsilon/2^{cm}$ for some constants $c$ and $\epsilon$. Now, we use $\mathcal{O}$ to give a $\BPP$ algorithm to solve a $\cCP$ hard problem. 

Pick a $\cCP$ hard function $f(x)$ and construct the circuit $C$ as in \thm{ccp_QAOA} for QAOA, or \thm{ccp_IQP} for IQP or \cite[Lemma 10]{bouland2021noise} for Haar random circuits. From these theorems, we have that additively approximating the output probability $p(C)$ of that circuit to within an error of $1/2^{2n+1}$ is $\cCP$ hard. Next, construct an interpolation $C(\theta)$ from this circuit to a random circuit and, as before, we will denote the probability of observing the outcome $0^n$ for the circuit $C(\theta)$ by $p(\theta)$.

Using \thm{TVD}, we have that the distribution of the circuits $C(\theta)$ and $C(0)$ are at most $O(1)$-close in total variation distance when $\theta\leq \Delta<1$. Since the classical algorithm $\mathcal{O}$ computes $p(\tilde{C})$ as given in \eq{O_prob}, using \lem{random_close}, the classical algorithm $\mathcal{O}$ also computes $p(\theta)$ to the same additive error with a constant probability of at least $1-\eta'$ (for some constant $\eta'$) for $\theta\leq\Delta$ i.e.,
\begin{equation}
    \pr[|\mathcal{O}(C(\theta)) - p(\theta)|\geq\delta]\leq \eta'\,.
\end{equation}
Now, we choose $O(d\log d)$ points $\theta_i$ between $[0,\Delta]$ such that the points $(2\theta_i/\Delta)-1$ are Chebyshev distributed in the interval $[-1,1]$. Let $y_i=\mathcal{O}(C(\theta_i))$ be $i^{th}$ sampled outcome. By our assumption on $\mathcal{O}$ and \lem{random_close}, the outcomes $y_i$ satisfy
\begin{equation}
    \pr[|y_i-p(\theta_i)|\geq\delta]\leq \eta'\,.
\end{equation}
Now, using \thm{polynomial}, the probability $p(\theta)$ for $\theta\in[0,\Delta]$ satisfies
\begin{equation}
    |p(\theta)-\tilde{p}(\theta)|\leq \frac{1}{2^{2n+2}}\leq \delta\,,
\end{equation}
where $\tilde{p}(\theta)$ is a degree $d$ polynomial, where
\begin{equation}
    d=O(\frac{m}{\log m})\,.
\end{equation}
Putting the above two together, for $\theta\in [0,\Delta]$, we have that
\begin{align}\label{eq:O_p_tilde}
    \pr[|y_i-\tilde{p}(\theta_i)|\geq\delta]&\leq \pr[|y_i-p(\theta_i)|\geq\delta] \nonumber\\
    &+ \pr[|p(\theta_i)-\tilde{p}(\theta_i)|\geq\delta] \nonumber\\
    &\leq \eta'\,.
\end{align}
This means that the set of points $D=(\theta_i,y_i)$ satisfy the conditions of \thm{RBW_P} and now using that theorem, we can get an additive approximation $\hat{p}(m)$ to $\tilde{p}(m)$ such that 
\begin{equation}\label{eq:p_hat_p_tilde}
    \pr[|\hat{p}(m) - \tilde{p}(m)|\geq \epsilon']\leq \frac{1}{3}\,,
\end{equation}
where the error is 
\begin{equation}
    \epsilon'= \frac{9\delta}{4\Delta^{'d}}= \frac{1}{2^{2n+2}}\,,
\end{equation}
where we used \eq{delta} for the second equality.
By \thm{polynomial} again, we have
\begin{equation}
    p(m)=\tilde{p}(m)\,.
\end{equation}
Combining the above equation and \eq{p_hat_p_tilde}, we get that
\begin{equation}
    \pr[|\hat{p}(m) - p(m)|\geq \epsilon']=\pr[|\hat{p}(m) - \tilde{p}(m)|\geq \epsilon']\leq \frac{1}{3}\,,
\end{equation}
This means that when $p(m)=0$, this gives us w.p $2/3$
\begin{equation}
    \hat{p}(m)< \frac{1}{2^{2n+2}}\,.
\end{equation}
When $p(m)\geq 1/2^{2n}$, then w.p $2/3$
\begin{equation}
    \hat{p}(m)\geq  p(m) - \frac{1}{2^{2n+2}}\geq \frac{3}{2^{2n+2}}\,. 
\end{equation}
Therefore, using the value of $\hat{p}(m)$, we can decide between the two cases. The result follows since deciding whether $p(C)=0$ or $p(C)\geq 1/2^{2n}$ is $\cCP$ hard.
\end{proof}
For QAOA and IQP circuits, we have the following improved dependence on the error.
\begin{corollary}
For QAOA and IQP circuits we have $\cCP$ hardness for an additive error of $2^{cn}$ for some constant $c$.
\end{corollary}
\begin{proof}
For QAOA and IQP circuits, using \cor{QAOA_IQP} the above proof goes through if we replace $d$ by
\begin{equation}
    d=O(\frac{n}{\log n})\,,
\end{equation}
and use $m=\poly(n)\leq n^{O(1)}$. For this range of values of $d$ and $m$, we have $\delta\leq \epsilon/2^{cn}$.
\end{proof}

\section{Conclusions}\label{sec:conclusions}
In this work, we have shown that three random circuit families, namely $p=1$ QAOA, Haar random and random IQP circuits have output probabilities that are $\cCP$ hard to approximate classically in the average-case. This means that unless the polynomial hierarchy collapses, it is not possible for classical algorithms to efficiently approximate the output probability of random circuits with high probability. The error to which classical hardness applies scales as $2^{-O(n)}$ for QAOA and IQP circuits. For Haar random circuits, this scaling is $2^{-O(m)}$, which improves on earlier results \cite{bouland2021noise, KMM21} where $2^{-O(m\log m)}$ scaling is shown.

As mentioned earlier, moving beyond this error scaling to $O(2^{-n})$, which would imply hardness of approximate \emph{sampling}, would require new ideas such as depth sensitive polynomial interpolation. This is because in \cite{Napp_et_al}, it was shown that output probabilities of Haar random circuits at depth 3 can be efficiently classically simulated. Even if one is able to design depth sensitive polynomial interpolation, proving average-case hardness via worst-case $\cCP$ hardness presents additional challenges. The error scaling that $\cCP$ hardness already introduces is $2^{-2n}$, which is already worse than the one required for sampling (using Stockmeyer's theorem).

It would be interesting to see if anti-concentration is true for $p=1$ QAOA and IQP circuits which could be useful to be able to prove hardness of sampling. For Haar random circuits, it was shown in \cite{Hangleiter2018anticoncentration, harrow_mehraban,PRXQuantum.3.010333}. Anti-concentration would follow if the distributions used here for QAOA and IQP are 2-designs. For a different distribution on QAOA where the full $X$ basis unitary on the $n$ qubits is also picked to have a uniformly random phase, it was shown in \cite{X_Z_diagonal_t_designs} that for large enough $p$, the distribution is a $t$-design, which implies that this distribution anti-concentrates.

Finally, it will be of practical importance to consider the effects of noise on random circuit sampling of QAOA circuits. Several papers have considered the effect of noise on classical hardness for Haar random circuits, such as the recent work \cite{Noh2020efficientclassical, bouland2021noise}. For IQP circuits, in \cite{Bremner2017achievingquantum}, a the hardness results are shown to be robust to a small amount of noise at the end of the circuit using classical error correction.

\section*{Acknowledgements}
I would like to thank Luke Govia, Mark Saffman, Tom Noel, Cody Poole and the ONISQ Cold Quanta team for useful discussions on QAOA. I am also grateful to Bill Fefferman for his comments on an earlier draft. This material is based on work supported by the Defense Advanced Research Projects Agency (DARPA) under Agreement No. HR001120C0068.
\printbibliography
\end{document}